\documentclass[11pt, preprint]{imsart} 
\usepackage[a4paper, left=30mm,right=30mm,top=30mm,bottom=30mm]{geometry}  
              
\usepackage{amsmath,amssymb,amsfonts,amsthm,graphicx} 
\usepackage[sort&compress,round,authoryear]{natbib}
\usepackage[colorlinks,citecolor=blue,urlcolor=blue]{hyperref}
\usepackage{url} 
\usepackage[utf8]{inputenc}
\numberwithin{equation}{section}
\usepackage{type1cm}
\usepackage{float}
\RequirePackage{graphicx}

\startlocaldefs
\newcommand{\n}[0]{\hspace*{.35em}}
\newcommand{\nn}[0]{\hspace*{.7em}}
\newcommand{\ful}{\mbox{$\, \frac{ \nn \nn \;}{ \nn \nn
}$}}

\renewcommand{\hat}[1]{\widehat{#1 }}
\renewcommand{\tilde}[1]{\widetilde{#1 }}

\newcommand{\ci}{\mbox{\protect $\: \perp \hspace{-2.3ex}\perp$ }}
\newcommand{\pr}{\mathrm{Pr}}
\renewcommand{\b}[1]{\boldsymbol{#1}}

\newcommand{\half}{\ensuremath{\tfrac{1}{2}}}
\newcommand{\reals}{\mathbb{R}}

\renewcommand{\arraystretch}{1.1}

\newtheorem{thm}{Theorem}[section]

\newtheorem{prop}[thm]{Proposition}
\newtheorem{cor}[thm]{Corollary}

\theoremstyle{definition}

\newtheorem{exa}[thm]{Example}

\theoremstyle{remark}
\newtheorem{rem}[thm]{Remark}

\newcommand{\hcal}{\ensuremath{\mathcal{H}}}

\newcommand{\inner}{\boldsymbol{\cdot}}
\newcommand{\had}{\ensuremath{\scriptsize\begin{pmatrix} 
1 & \phantom{-}1 \\ 
1 & -1\\ 
\end{pmatrix}}}
\newcommand{\fracshalf}{\mbox{$\frac{1}{2}$}}
\endlocaldefs
\begin{document}

\begin{frontmatter}

\title{Palindromic Bernoulli distributions}
\runtitle{Palindromic Bernoulli distributions}
\begin{aug}
  \author{\fnms{Giovanni M.} \snm{Marchetti}\corref{}\ead[label=e1]{giovanni.marchetti@disia.unifi.it}}
  \address{Dipartimento di Statistica, Informatica, Applicazioni ``G. Parenti'', Florence, Italy\\\printead{e1}}
  \and
  \author{\fnms{Nanny} \snm{Wermuth}\ead[label=e2]{wermuth@chalmers.se}}\address{Mathematical Statistics, Chalmers University of Technology, Gothenburg, Sweden\\ and Medical Psychology and Medical Sociology, Gutenberg University Mainz, Germany\\\printead{e2}} 
\end{aug}         
\runauthor{G. M. Marchetti and N. Wermuth}

\begin{abstract}
We introduce and  study  a subclass of  joint Bernoulli distributions which has the palindromic property.  For such distributions the vector of joint probabilities is unchanged when  the  order of the elements is reversed. We prove 
 for binary variables that the palindromic property is equivalent  to zero constraints on all odd-order interaction parameters, be it in parameterizations
which are log-linear, linear or  multivariate logistic. In particular, we  derive the  one-to-one parametric transformations for  these three types of model specifications and give simple closed forms  of  maximum likelihood estimates.  Several special cases are discussed and a case study is described.
\end{abstract}



\tableofcontents

\end{frontmatter}


\section{Introduction}\label{sec:intro}

A sequence of characters, such as  \texttt{QR*-TS}, becomes a  palindromic sequence when the order of the characters is reversed and appended,  here to give  
\texttt{QR*-TSST-*RQ}.  The notion is used in somewhat modified forms, among others, in musicology, biology and linguistics.  An example of a palindromic  sentence which respects the spacings between words is `step on no pets'.

Here, we adapt the term to  Bernoulli distributions. For a single binary variable, the distribution is palindromic if it is uniform, that is if both levels occur with probability $1/2$. For a Bernoulli distribution of $d$ binary variables $A_1, \dots, A_d$, having a probability mass function $p(a)$ with $a$ in the set of all binary $d$-vectors,  the distribution is palindromic if 
$p(a) = p(\sim a) \quad \text{for all } a$,     
where $\sim a$  is  the complement of  $a$; for instance,  $\sim a = (0,1,0)$  for $a = (1,0,1)$. 

 With $\alpha, \beta, \gamma, \delta$ denoting probabilities, bivariate and  trivariate palindromic Bernoulli distributions
can be written, as in the following tables:

\begin{small}
$$
\begin{tabular}{lcccc}
\hline\\[-5mm]
$A_1$    &\hspace{-6mm}  $A_2\!: $\hspace{-5mm} &   0   & 1               & \text{sum}\\[1mm] \hline
\, 0    && $\alpha$         &$ \beta$ & $1/2$ \\ 
\, 1        &  &$ \beta$ &$ \alpha  $             & 1/2 \\[1mm] \hline
\text{sum}& & 1/2     & 1/2           & 1\\[1mm]  
\hline 
\end{tabular} \hspace{10mm}
\begin{tabular}{lcccccc}
\hline
    \quad         &\hspace{-6mm}  $A_3\!:$\hspace{-5mm}  &0   & 0         & 1  & 1  &           \\
  $A_1 $         &\hspace{-6mm}  $A_2\!:$\hspace{-5mm}  &0   & 1         & 0  & 1  & \text{sum}\\ \hline
\, 0    & &$\alpha $         & $\gamma  $       &$ \delta $  &$  \beta$   & 1/2    
\\ 
\, 1         & &  $\beta $        &$ \delta  $       & $\gamma $  &$ \alpha $    & 1/2   
\\ \hline
\text{sum} &&$ \alpha+\beta$       &  $\gamma + \delta$     &$ \gamma + \delta$&
$\alpha + \beta$  & 1         \\   \hline
\end{tabular}
$$  
\end{small} 
 
 Continuous distributions may also be palindromic. This concept extends the discussion of  the diverse forms of multivariate symmetry by  \citet{serfling2006}  since  it operationalizes his notion of central symmetry  in an attractive way. Let  the above  binary vector $a$  define for  $d$ mean-centred  continuous variables  their  orthant probabilities. 
 Such a  distribution is  palindromic if the  probabilities of the $d(a)$-orthant and the  $d(\sim a)$-orthant  coincide for all  $a$.  
Examples are for instance  mean-centred Gaussian  and   spherical distributions.  
\begin{figure}[h]
\includegraphics[scale = .5]{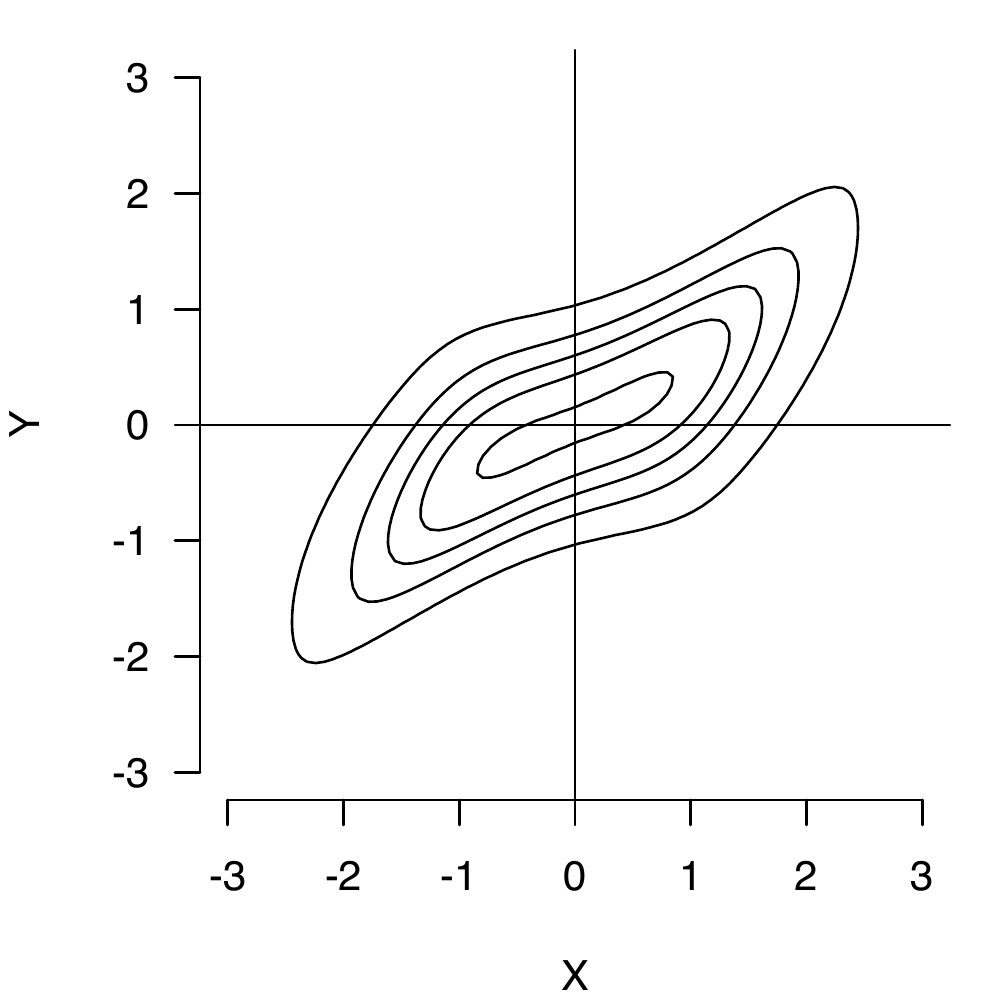} 
\hspace{1cm}
\includegraphics[scale =.5]{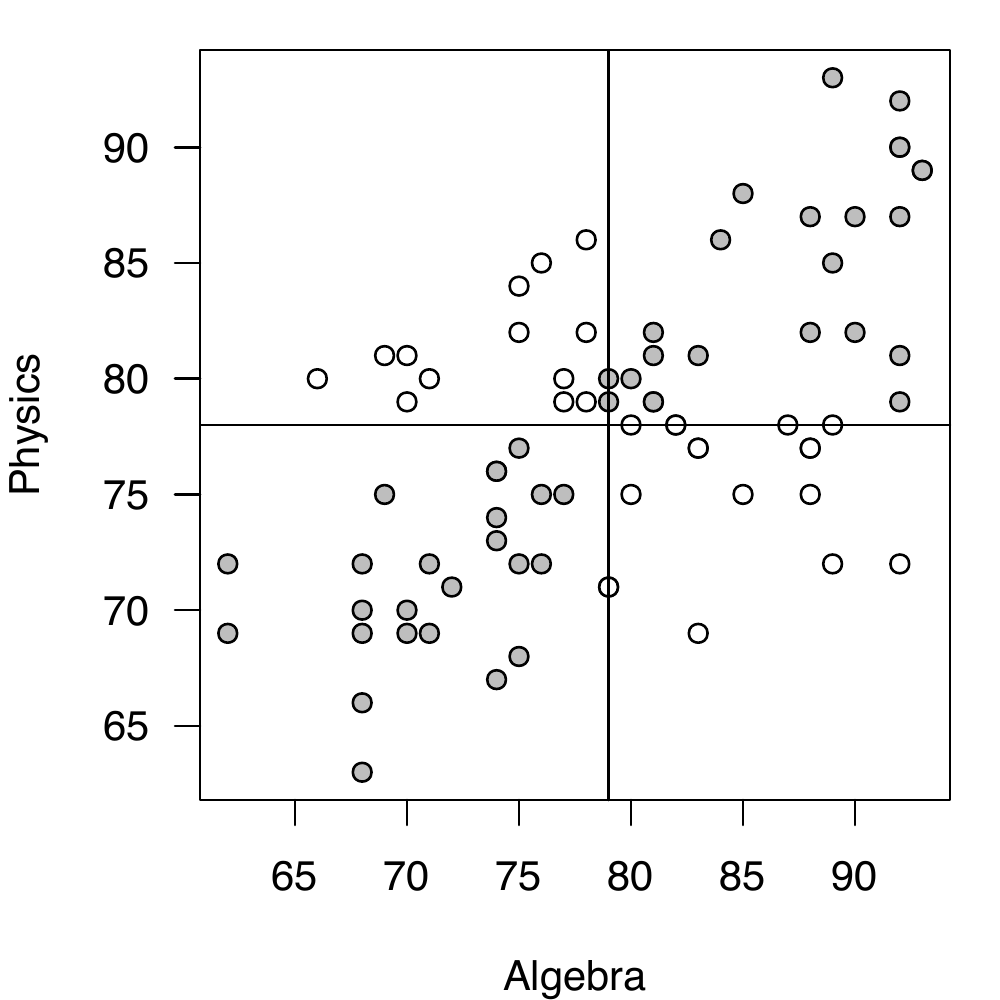} 
\caption{\label{fig:cs} Contour levels of a centrally symmetric bivariate density (left) and a median dichotomized sample (right). The data are from the case study of section~\ref{sec:case}.}
\end{figure}

 For general discrete variables,  the palindromic property differs from  \emph{complete symmetry} defined by  \citet{bhapkar1990}. Complete symmetric tables satisfy the condition  $p(a) = p(\sigma(a))$, for any $a$ and for any permutation $\sigma$  of the indices.  
     \citet[app. C]{edwards2000} median-dichotomized joint Gaussian distributions  and  proved that the resulting binary  probabilities  give a non-hierarchical log-linear model in which all odd-order interactions vanish, that is  all terms involving an odd number of factors are zero.

In this paper,  we study properties of palindromic Bernoulli distributions in general. In particular, we prove that the vanishing of all  odd-order log-linear interactions  is not only a necessary but also a sufficient condition. We show the same characterization for  models  linear-in-probabilities,  \citet{streitberg1990}, and for the multivariate logistic parametrization, \citet{glonek1995}, and explain why palindromic Bernoulli distributions
with Markov structure are in the regular exponential family.

\section{Characterization in terms of interaction parameters} \label{sec:charact}

In this section we introduce three types of parameterization and we show that the palindromic Bernoulli distributions can be characterized by the vanishing of all odd-order interactions, no matter whether these are log-linear, linear or multivariate logistic terms.

\subsection{Notation}
Let  $A = (A_1, \dots, A_d)$ be a random vector with a \emph{multivariate
Bernoulli} distribution. Thus, $A$ takes values $a = (a_1, \dots, a_d)$ in the
set $\mathcal{I} = \{0,1\}^d$ with  probabilities 
$$
 p(a) = \pr(A_1 = a_1, \dots, A_p = a_p), \quad  \textstyle{\sum_{a\in \mathcal{I}}} \,p(a) = 1.
$$  
For simplicity, we assume  $p(a) >0$ for all $a$.
The probability distribution of $A$ is determined by the  $2^d \times 1$ vector
$  \pi$ containing all the probabilities $p(a)$ and belonging to the
$(2^{d}-1)$-dimensional simplex. 
We  list vectors $a$ in a lexicographic order such that 
 the first index in $a$ runs  fastest, then the second changes and  the last index runs slowest.
 Cells of a corresponding contingency table are in vector $b\in \mathcal{I}$.

Given a subset $M \subseteq V$ of the variables,   the  marginal distribution of
the  variables $A_v$, for $v \in M$ has itself a joint Bernoulli
distribution, in the same lexicographic order: 
$$
p_M(a_M) = \pr(A_v = a_v, \text{ for all } v \in M).
$$

We  use three well-studied parameterizations for  joint Bernoulli distributions, that is the log-linear, the linear and the
multivariate logistic parameterizations and show how and why they differ even
for palindromic Bernoulli distributions.

In general, a parameterization of $A$ is a smooth one-to-one transformation,
mapping  $\pi$ into a  $2^p \times 1$ vector $  \theta = G(  \pi)$, say, 
whose entries $\theta_b$, are called \emph{interaction parameters}. 
To index interactions, it is useful to have a one-to-one mapping between the
cells in $b$ and subsets of  $V = \{1, \dots, p\}$. For 
$p = 3$:
$$
\small
\begin{array}{rcccccccc} \hline\\[-6mm]
            & \multicolumn{8}{c}{\text{Lexicographic order}} \\
            \cline{2-9}\\[-5mm]
\text{cells in } b: & 000 & 100 & 010 & 110 & 001 & 101 & 011 & 111\\
 \text{ subset of $V$:} & \emptyset & 1 & 2 & 12 & 3 & 13 & 23 & 123\\
\theta_b:       & \theta_\emptyset & \theta_1 & \theta_2 & \theta_{12} &
\theta_3 & \theta_{13} & \theta_{23} & \theta_{123}\\ \hline\\[-6mm]
\end{array}
$$
The cardinality of the set $b$, denoted by $|b| = \sum_v b_v$,
gives the number of ones in vector $b$.
Depending on  $|b|$ being odd or even, an interaction parameter $\theta_b$  is said to be of odd or even
order.  For instance, the even-order
$\theta_{13}$ is a two-factor interaction of $A_1$ and  $A_3$.

\subsection{Log-linear parameters}
 \emph{Log-linear  parameters}
  are  contrasts of log probabilities, that is   linear combinations of  log $p(a)$, with
weights adding to zero. 
The vector of the log-linear parameters is  
\begin{gather}
  \lambda  =  \hcal_d^{-1} \log   \pi  \label{eq:lambda} 
\end{gather}
where 
$$
\hcal_d = \underbrace{\had \otimes \cdots \otimes \had}_{d}.
$$
is a $2^d\times 2^d$  symmetric design matrix whose  generic entry is $h_{ab} = (-1)^{a \inner b}$ for 
$(a,b)\in \mathcal{I}\times \mathcal{I}$. Its inverse, the contrast matrix, is
$\hcal_d^{-1} = 2^{-d}\mathcal{H}_d$.  The 
special form of $\hcal_d$, chosen here,  uses so-called \emph{effect coding}; see for instance
\citet{wercox1992}. The individual interactions can be written as
\begin{equation}
 \lambda_b = 2^{-d} {\textstyle{\sum_{a \in \mathcal{I}_V}}}  (-1)^{a \inner b} \log
p(a),\label{invloglini}
\end{equation}
where  $a \inner b = a_1 b_1 + \cdots + a_p b_p$ is the inner product of 
the two binary vectors $a$ and $b$; see \citet[p. 619]{haberman1973}. In equation \eqref{invloglini}, the symbol $b$ is
interpreted for $\lambda_b$ as a subset of $V$   
and  in the expression $(-1)^{a \inner b}$ as binary vector.

 The inverse mappings from $  \lambda
$ to $  \pi$ may be explicitly computed as 
\begin{equation}
  \pi  =  \exp(\hcal_d   \lambda)  \label{eq:pifromlambda}, \quad \quad
p(a) = \exp\big [\textstyle \sum_{b \in \mathcal{I}}   (-1)^{a\inner
b}\lambda_b \big ]. 
\end{equation}

Bernoulli distributions with positive cell probabilities  belong to the so-called regular exponential family with the vector $  \lambda$ containing the canonical parameters. 
\subsection{Linear interactions or moment parameters}

In contrast to log-linear models, the linear-in-probability models,  discussed for instance  by \citet[app. 2]{coxwer1992},  and their interactions are based on moments.
The vector 
$  \xi =  \hcal_d   \pi$  is  a \emph{moment parameter vector} and the mapping
between $  \xi$  and $  \lambda$ is one-to-one and differentiable; see
\citet[ p.~121]{nielsen1978}.

The moment vector $  \xi$ is proportional to the expected value of the
sufficient statistics for $  \lambda$. 
With $  y$ denoting the  vector of the frequencies
and by  the symmetry of $\hcal_d$ of equation \eqref{eq:lambda}, this vector  of sufficient statistics  is
 $\hcal_d   y$. The elements
of $  \xi$, called also \emph{linear interactions}, are  
\begin{gather}
\xi_b =\textstyle{ \sum_{a \in \mathcal{I}}}  (-1)^{a \inner b} \, p(a)
\label{invlini}
\end{gather}
and  gives as  inverse transformations 
\begin{equation}
  \pi = 2^{-d}\; \hcal_d \,  \xi, \quad  \quad p(a) = 2^{-d} \textstyle{\sum_{b \in \mathcal{I}} } (-1)^{a
\inner b} \xi_b.
 \label{eq:lin}   
\end{equation}
For instance, with $d = 2$ the two-factor linear interaction is
$\xi_{11} = p_{00}-p_{01}-p_{10}+p_{11}$.

 For  the
transformed random variables $D_v = (-1)^{A_v}$, which take value $1$ if $A_v = 0$
and $-1$ if $A_v = 1$, the individual interactions are
\begin{gather}
\xi_b = E\Big(\textstyle{ \prod_{v \in V}} D_v^{b_v}\Big) \label{eq:mix}.
\end{gather} 
 Because the element $(-1)^{a\inner b} $ in equation~\eqref{invlini} may be
written with  $d_v = (-1)^{a_v}$ as 
$$
(-1)^{a\inner b} = (-1)^{a_1b_1}\times \cdots \times (-1)^{a_pb_p}
 = \textstyle \prod_{v\in V} (d_v)^{b_v} ,
$$
equation \eqref{invlini} gives the expected value 
of this product with respect to $p(a)$.

Equation \eqref{eq:mix} implies that each moment parameter,
$\xi_b$,  is a \emph{marginal parameter}, defined in
 the marginal distribution 
of the random vector $(A_v)_{v \in b}$, while the log-linear parameter
$\lambda_b$  is defined in the joint distribution.
Therefore there is, in general, no simple relation between  
     the log-linear parameter  $\lambda^M_b$, say, in the marginal distribution
$p_M(a_M)$  and  $\lambda_b$, the log-linear parameter in the joint
distribution, but there are exceptions, see Example~\ref{ex:orth} and Section \ref{sec:ising}. By contrast, the 
moment vector $\xi_b^M$ defined in $p_M(a_M)$ coincides with 
$\xi_b$.  
  
According to an important result for regular exponential families by \citet[pp.~121--122]{nielsen1978}, for an arbitrary partition of the
parameter vectors 
$  \lambda$ and $  \xi$ in two  sub-vectors  such that
$  \lambda = (  \lambda_\mathcal{A},   \lambda_\mathcal{B})$ and $  \xi =
(  \xi_\mathcal{A},   \xi_\mathcal{B})$, the distribution $  \pi$ is
uniquely parameterized by the mixed vector
$(  \lambda_\mathcal{A},   \xi_\mathcal{B})$ or  by 
$(  \xi_\mathcal{A},   \lambda_\mathcal{B})$ and 
there is a diffeomorphism between this \emph{mixed parameterization} and the
log-linear parameter $  \lambda$  or the moment parameter 
$\xi$.
 
%

 \subsection{Multivariate logistic parameters}
The \emph{multivariate logistic parametrization}, introduced by 
\cite{glonek1995}, is defined by the highest order log-linear parameters, 
considered here under effect coding,  in each possible marginal distribution of $A$.
The parameters are given by the vector 
$  \eta = (\eta_b)_{b \subseteq V}$ where 
\begin{equation}     
 \eta_b =  \lambda^b_b, \quad \quad b \subseteq V. \label{eq:mlogit}
\end{equation}

\citet{kauermann1997} showed 
that the mapping $T:   \lambda \mapsto   \eta$  
from the log-linear to the multivariate logistic parameters is a diffeomorphism
by proving that $T$ is a composition of smooth transformations between the
canonical, the moment and the mixed parameters. More details are given below in
subsection~\ref{sec:pro}.

Let $\Lambda= \mathbb{R}^{2^{p}-1}$ be the parameter space for the log-linear
parameters $  \lambda$. Then the parameter space $E =
T(\Lambda)$ for 
$  \eta$ is the image of $  \lambda$ under the transformation $T$. 
 Explicit
forms for the inverse function $T^{-1}: E \rightarrow \Lambda$ are  known for 
$p = 1$ or $p = 2$ and in special cases, such as in Example~\ref{ex:orth}. An
algorithm  provided by \cite{qaqish2006} detects simultaneously  whether
the vector $  \eta$ is compatible with a proper probability vector $  \pi$.


%
%

\subsection{Properties of the palindromic Bernoulli distributions}\label{sec:pro}
We now study several properties of palindromic Bernoulli distributions 
and start by proving that these distributions are closed under marginalization.

\begin{prop}\label{marginalizePBD}
If $ p(a)$ is a palindromic  Bernoulli 
distribution then,  for any subset $M$ of the variables, the marginal distribution $p_M(a_M)$ is  palindromic.    
\end{prop}
 \begin{proof}
Define the partition $a = (a_N,a_M)$ and let 
 $ p_M(a_M) = \textstyle{\sum_{a_N \in \{0,1\}^{|N|}}}\, p(a_N, a_M).$
 Then if the distribution is palindromic, $ p(a_N, a_M) = p(\sim a_N, \sim
a_M)$ and
\begin{gather*}
  p_M(a_M) =\textstyle{ \sum_{\sim a_N \in \{0,1\}^{|N|}}}\, p(\sim a_N, \sim a_M) = p_M(\sim a_M).    \qedhere
\end{gather*}
  \end{proof}
Next, we  characterise the distribution by  zero constraints on interactions. 
\begin{prop}\label{char1}
A Bernoulli distribution  is palindromic if and only if,  with $\theta_b = \xi_b$  or $\theta_b = \lambda_b$, all 
odd-order linear or log-linear interactions  vanish, that is  if and only if 
$$
\theta_b  = 0,\quad  \text{ for all }  b \subseteq V \text{  with } |b| \text{
odd}.
$$ 
\end{prop} 

\begin{proof} 
1) (If $A$ is palindromic then all odd-order $\xi_b = 0$.)
Any linear interaction can be written as  
\begin{equation}
 \xi_b = {\textstyle\sum_{ a \in \mathcal{I}_1}}  (-1)^{a \inner b}\,p(a) 
 + \textstyle{\sum_{a \in \mathcal{I}_1}} (-1)^{(\sim a) \inner b }\,  p(\sim
a) ,
\end{equation} 
where $\mathcal{I}_1$ denotes the subset of  cells having a one as first
element.
Thus $\mathcal{I}_1$ contains half of the cells.
If the distribution is palindromic, $p(\sim a) = p(a)$ and $(-1)^{(\sim a)
\inner b)} = (-1)^{|b|} (-1)^{a \inner b}$. Thus,    
\begin{equation}
 \xi_b = {\textstyle\sum_{a \in \mathcal{I}_1}}  (-1)^{a \inner b}\,p(a)
 +(-1)^{|b|}   {\textstyle \sum_{a \in \mathcal{I}_1}}  (-1)^{a \inner b}p(a).
\end{equation} 
When $|b|$ is odd then $(-1)^{|b|} = -1$ and $\xi_b = 0$;  see also  \citet[App. C]{edwards2000}.

2) (If all odd-order $\xi_b=0$, then $p(\sim a) = p(a)$.) 
If all odd-order interactions vanish, then
\begin{equation}
p(a) =  \frac{1}{2^d}   \textstyle{\sum_{b \in \mathcal{I}_\mathrm{even}}}
(-1)^{a \inner b}\, \xi_b, 
\end{equation}
where $\mathcal{I}_\mathrm{even}$ is the subset of the cells $b$ such that  
 $|b|$ is even. Thus,   
\begin{equation}
p(\sim a) =\frac{1}{2^p} {\textstyle  \sum_{b \in \mathcal{I}_\mathrm{even}}} 
(-1)^{(\sim a)\inner b}\, \xi_b 
= \frac{1}{2^p}   {\textstyle  \sum_{b \in \mathcal{I}_\mathrm{even}}}   (-1)^{
|b|}  (-1)^{a \cdot b} \,\xi_b =p(a),
\end{equation}
because $|b|$ is even. So the distribution is palindromic. 

3) The same arguments apply for the log-linear parameterization. The
distribution is palindromic if and only if  
     $\log p(a) = \log p(\sim a)$ for all $a$. Therefore, using equation 
      \eqref{invloglini}, and the previous lines of reasoning,   
      $\lambda_b = 0$ whenever $|b|$ is odd. Conversely, 
if all odd-order log-linear parameters  $\lambda_b$ vanish, then 
from
$$
\log p(a) =\textstyle \sum_{b \in \mathcal{I}}   (-1)^{a\inner b}\lambda_b
$$
we get  $\log p(a) = \log p(\sim a)$ and the distribution is palindromic.
\end{proof}
By equation~\eqref{eq:mix} and Proposition~\ref{char1}, the joint distribution
of $A$ is palindromic if and only if all the odd-order 
moments of $D = (-1)^A$ are zero. Also, as the palindromic property is characterized by linear constraints on the canonical pararameters $\lambda$,  we have the following result.
\begin{cor}
	Palindromic Bernoulli  distributions are  a regular  exponential family.
\end{cor}

 We show next a similar characterization for the multivariate 
logistic parametrization. 
\begin{prop}\label{char2}
A Bernoulli distribution  is   palindromic 
if and only if all  odd-order multivariate logistic parameters vanish, that is  if
and only if 
$$
\eta_b  = 0,\quad  \text{ for all }  b \subseteq V \text{  with } |b| \text{
odd}.
$$ 
\end{prop}
The analogous result for the larger class of complete hierarchical marginal log-linear parameterizations, \cite{BerRud2002}, will be discussed elsewhere.

 The proof uses  a  transformation $T\colon \lambda
\mapsto   \eta$ introduced by 
\citet[p.~265]{kauermann1997}. The composition of smooth one-to-one
transformations $T_M$ gives  $T$, for each nonempty subset $M \subseteq V$. 
The functions $T_M$ operate on parameter transformations between the canonical and the moment parametrizations, as follows.
If $M = V$,
$$
T_M(  \lambda_{\mathcal{P}(V)\setminus V)}, \lambda_V) = 
(  \xi_{\mathcal{P}(V)\setminus V)}, \lambda_V).
$$ 
If  $M\subset V, |M| \ne 1$:
$$
T_M(\dots,   \xi_{\mathcal{P}(V)\setminus V)}, \xi_M, \dots)= (\dots, 
\xi_{\mathcal{P}(M)\setminus M)}, \eta_M, \dots)
$$ 
and the remaining parameters, which are not listed, are  left unchanged.
Finally, if $|M| = 1$, 
$$
T_M(\xi_M, \dots) = (\eta_M, \dots).
$$

For instance, to clarify, to get $  \eta = T(  \lambda)$ for three variables we define
$$
T(  \lambda) = T_1 \circ T_2 \circ T_3 \circ T_{12} \circ T_{13} \circ T_{23} \circ T_{123}(  \eta)
$$
and Table~\ref{tab:tm} gives the details of the  required transformations $T_M$.
 \begin{table}
 \centering
 \caption{The sequence of transformations  required to obtain the multivariate logistic parameter $ \eta$ from the log-linear parameter $ \lambda$. }   
  \label{tab:tm}
 $$
 \begin{array}{ccccccccc}\hline
 \text{Transformation} & \multicolumn{7}{c}{\text{Parameters}}& \text{Intermediate result}\\ \hline
 T_{123}( \lambda) & \xi_1& \xi_2 &  \xi_{12}& \xi_{3}& \xi_{13}& \xi_{23}& \eta_{123} &  \theta^{(1)}	\\
  T_{23}( \theta^{(1)}) & \xi_1& \xi_2& \xi_{12}& \xi_{3}& \xi_{13}& \eta_{23}& \eta_{123} &  \theta^{(2)}\\  
    T_{13}( \theta^{(2)}) & \xi_1& \xi_2& \xi_{12}& \xi_{3}& \eta_{13}& \eta_{23}& \eta_{123} &  \theta^{(3)}\\ 
        T_{12}( \theta^{(3)}) & \xi_1& \xi_2& \eta_{12}& \xi_{3}& \eta_{13}& \eta_{23}& \eta_{123} &  \theta^{(4)}\\ 
      T_{3}( \theta^{(4)}) & \xi_1& \xi_2& \eta_{12}& \eta_{3}& \eta_{13}& \eta_{23}& \eta_{123} &  \theta^{(5)}\\  
   T_{2}( \theta^{(5)}) & \xi_1& \eta_2& \eta_{12}& \eta_{3}& \eta_{13}& \eta_{23}& \eta_{123} &  \theta^{(6)}\\    
    T_{1}( \theta^{(6)}) & \eta_1& \eta_2& \eta_{12}& \eta_{3}& \eta_{13}& \eta_{23}& \eta_{123} &  \eta\\ [1mm]   \hline     
 \end{array}
 $$ 
\end{table}  
 
\begin{proof}
Let ${\rm odd} = \{b \in \{0,1\}^{|V|}: |b| \;{ \rm  odd}\}$ be the subset of
all odd-order interactions. 
Then, below we show that $  \eta_{\rm odd} =   0$ if and only if $\b
\lambda_{\rm odd} =   0$.

From Proposition~\ref{char1} we know that a binary distribution is palindromic  if
and only if $  \lambda_{\rm odd} =   0$. If $  \pi$ is palindromic then all
the marginal distributions $p_b(a_b), b \subseteq V$ are palindromic and thus
in each of them, such that $|b|$ is odd,   $\lambda_b^b = \eta_b = 0$. Thus,
$  \eta_{\rm odd} =   0$. Let ${\rm even} = \mathcal{P}(V) \setminus {\rm
odd}.$ Then 
\begin{equation}
    T(  \lambda_{\rm even},   \lambda_{\rm odd} =   0) = 
    (  \eta_{\rm even},   \eta_{\rm odd} =   0). \label{eq:kau}
\end{equation}
Conversely, if $  \eta_{\rm odd} =   0$, let $   \eta_{\rm even}$ be 
arbitrarily chosen such as $  (  \eta_{\rm even},   \eta_{\rm odd}) =   0)
\in E_0 \subset E$ (the parameter space of the $  \eta$s). As
$E_0$ is connected we can directly use equation \eqref{eq:kau} and the
smoothness of the inverse transformation $T^{-1}$ to get 
$$
T^{-1} (  \eta_{\rm even},   \eta_{\rm odd}=  0) = (  \lambda_{\rm even},
  \lambda_{\rm odd} =   0), 
$$ 
and thus the distribution $  \pi$ is palindromic.
\end{proof}
\noindent Table~\ref{tab:exa} illustrates  the different
parameters with a $2^3$  table. 
\begin{table}[htp!] 
\caption{Illustration of the  the different parameters with a $2^3$  table; constant terms  omitted.} \label{tab:exa}
$$ 
\begin{array}{rcccccccc} \hline
\text{cells }  b: & 000 & 100 & 010 & 110 & 001 & 101 & 011 & 111\\
80\,{ \pi}:  &  15& 9& 1&15& 15 &1 &9&15\\
\text{subsets of $V$ :} & \emptyset & 1 & 2 & 12 & 3 & 13 & 23 & 123\\
 \xi: & - &0&0& 1/2& 0 & -1/5&1/5&0\\
 \lambda: & - &0&0&\log (5)/2 & 0 & -1/5&1/5&0\\
  \eta: &-& 0&0& \log(3)/2  &0&  -\log(3)/2 & \log(3)/2
&0\\[1mm]
\hline
\end{array}
$$
\end{table}

Next, we state a result connected with 
  binary probability distributions \emph{generated by a linear triangular system}, as studied  in  \cite{wermarcox2009}. Their joint probabilities  
 may be defined by the recursive  factorization
$$
\pr(A_1 = a_1, \dots, A_d =a_d) =\pr(A_1= a_1) \textstyle{\prod_{s=2}^{d}} \pr(A_s = a_s \mid  A_{1} = a_{1}, \dots, A_{s-1} = a_{s-1})
$$
with uniform margins. With $\beta_{sj}$  denoting linear regression coefficients,
\begin{equation}
	 \pr(A_s = a_s \mid  A_{1} = a_{1}, \dots, A_{s-1} = a_{s-1}) = 
 \half(1 + \textstyle\sum_{1 = s-1}^d \beta_{sj}(-1)^{a_s + a_j} ).
 \label{eq:maineffects}
\end{equation}
The conditional expected values of $A_s$ given variables $A_{[s-1]} = (A_{1}, \dots, A_{s-1})$ are linear regressions  with only main effects and no constant term. For these distributions,  all the even-order linear interactions are  known  functions of  the marginal correlations. Here, we
 prove in Appendix~1  the following result and get back to such systems later.
\begin{prop}\label{prop:tri}
If   a binary probability distribution is generated with a linear triangular system, then it is palindromic. 
\end{prop}

\section{Independences and dependences}

Palindromic Bernoulli distributions share some but not all of the properties of joint Gaussian distributions.
We elaborate here on  properties of independences, of undirected dependences, also called
associations, and of directed dependences, also called effects. Conditional independence of variables $A, B$
given variable $O$, say, is written as  $A\ci B|O$, while the complement of it, called conditional dependence of $A, B$ given $O$,  is written as  $A\pitchfork B|O$; see \cite{wersad2012}. 

Starting with properties of general Bernoulli distributions, 
several measures of dependence are equivalent with respect to independences and the sign of a dependence; see \citet[thm. 1]{xiemageng2008}. The same happens for Gaussian distributions.
   In addition, for bivariate palindromic Bernoulli distributions, many  measures of dependence are even in  one-to-one correspondence; see e.g. \citet{WerMar2014}. 
   For the three variable table in Section~1,  conditional measures of main interest for $(A_1, A_2)$ at level 0  of $A_3$ are

 \begin{center}
 $(\alpha \delta)/ (\beta \gamma)$, the  odds-ratio,\\
$\{ \delta /(\gamma+\delta)\}  - \{ \beta/(\alpha +\beta) \}$,  the chance difference for success,  \\
$\{ \delta /(\gamma+\delta)\}  /\{ \beta/(\alpha +\beta) \}$,  the relative chance for success  and\\
$\rho_{12|k=0}=(\alpha \delta - \beta \gamma)\{ (\alpha +\beta) (\gamma+\delta)  (\alpha +\gamma)   (\beta+\delta)\}^{-\fracshalf}$,  the conditional correlation.
 \end{center}
 Thus, the chance difference for success and the conditional correlation $\rho_{12|k=0}$ are multiples of the cross-product difference, $\alpha \delta - \beta \gamma$, and coincide if and only if the
 probabilities for success are identical for $A_1$ and $A_2$.  Moreover, the odds-ratio is equal to the relative chance  if and only if $A_1 \ci A_2|A_3$, that is  if and only if $\alpha \delta=\beta \gamma$. A dependence is positive if $\rho_{12|k=0}$ (or the chance difference) is positive or the odds-ratio (or the relative chance) is $>1$.  
 
It can be verified that $\rho_{12|k=0} = \rho_{12|k=1}$,  and given the linear form of the conditional expectation $E(A_1, A_2 \mid A_3 = k) = \rho_{12|k}$ it follows from  \citet[thm. 1]{baba2004} that the constant conditional  correlation $\rho_{12|k}$ equals the partial correlation 
$$
\rho_{12.3} = (\rho_{12} - \rho_{13}\rho_{23})/c\quad  \text{ where } c = \{(1 - r_{13}^2)(1 - r_{23}^2)\}^{1/2}.
$$ 
 Special one-to-one relations among the mentioned different measure of dependence are given in the following Proposition, proved in Appendix 1.
 \begin{prop}\label{prop:triv}
 In a trivariate palindromic distribution:
 \begin{itemize}
\item[\rm (i)] $A_1 \ci  A_2 \mid A_3 \iff \lambda_{12} = 0 \iff \rho_{12|3} = 0 \iff \rho_{12.3} = 0$, 
\item[\rm (ii)] $A_1 \pitchfork  A_2 \mid A_2 > 0  \iff \lambda_{12} >  0 \iff \rho_{i12|3} > 0 \iff \rho_{12.3} > 0$, 
\item[\rm (iii)] $A_1 \pitchfork  A_2 \mid A_3 < 0 \iff \lambda_{12} < 0 \iff \rho_{12|3} < 0 \iff \rho_{12.3} < 0$. 
  \end{itemize}
  \end{prop}

 However, even if the chance difference is identical  at all level combinations of the remaining variables, the relative chance for success may vary widely. This shows for the palindromic distribution in Example 3.1,  where  the relative chance at level combination  $(0,0)$ of $A_3, A_4$ is more than 10 times higher than at $(1,1)$. This may become more extreme with equality just in  sign.
 $$
 \small
\begin{array}{r}
\hspace{-1mm} \text{\bf Example 3.1}\!\!\!\\[-1mm]
\\[-1.5mm]
 \end{array}
\left[\begin{array}{r r r r r r r r r r} 
 \text{levels  of  } A_1, A_2, A_3, A_4: & 0000 & 1000 & 0100 & 1100 &0010& 1010 & 0110 & 1110\\[-1mm]
9200 \; \pi:&4095& 91& 91& 47& 91 &47&47  &91\\[-1mm]
\end{array}\right]
$$
The explanation is the presence of a four-factor log-linear 
interaction in the $2^4$ table.

\subsection{Induced dependences and effect reversal} 
Next, we give three examples of binary palindromic  tables, which illustrate what have  been called the weak and the strong versions of the Yule-Simpson paradox.
In Examples 3.2 and 3.3,  an independence  gets destroyed by changing the conditioning sets  and in Example 3.4 the sign of a dependence gets reversed after marginalizing. The three examples illustrate in addition, that in all palindromic Bernoulli distribution not only  
conditional parameters are relevant but also the marginal parameters and, in particular,  
also simple correlations,  due to the one-to-one relation  between an odds-ratio and the correlation in their  bivariate distributions.
In the following examples the symbol $\lambda'$ indicates the log-linear interactions in effect-coding obtained from the counts. They are identical to previous  $\lambda$ parameters except for the 
constant term, where  $\lambda_{\emptyset} = \lambda'_\emptyset - \log n$.
$$\small
\begin{array}{r}
\text{\bf Example 3.2}\\[-1mm]
A\ci B|O\\[-1mm]
 \text{\& }  A\pitchfork B
 \end{array}
\left[\begin{array}{r r r r r r r r r r} 
 \text{levels } ijk \text{ of  } A,B,O: & 000 & 100 & 010 & 110 &001& 101 & 011 & 111\\[-1mm]
100 \; \pi:&32& 8& 8& 2& 2 &8&8 &32\\[-1mm]
\text{\!\!log-lin. interaction } \lambda':&\!\! 2.08& 0& 0& 0& 0& 0.69&0.69&0
\end{array}\right]
$$
$$
\small
\begin{array}{r}
\text{\bf Example 3.3}\\[-1mm]
A\pitchfork B|O  \\[-1mm]
 \text{\& }A\ci B
 \end{array}
\left[\begin{array}{r r r r r r r r r r}
 \text{levels } ijk \text{ of  } A,B,O: & 000 & 100 & 010 & 110 &001& 101 & 011 & 111\\[-1mm]
400 \; {\bf  \pi}:& 90 &60&40&10&10&40&60&90\\[-1mm]
\text{\!\!log-lin. interaction } {\lambda'}:&\!\! 3.65& 0& 0& -0.25& 0& 0.45&0.65&0
\end{array}\right]
$$
$$
\small
\begin{array}{r}
\text{\bf Example 3.4}\\[-1mm]
A\pitchfork B|O \text{ pos.}\\ [-1mm]
 \text{\& }A\pitchfork B \text{ neg.}
 \end{array}
\left[\begin{array}{r r r r r r r r r r}
 \text{levels } ijk \text{ of  } A,B,O: & 000 & 100 & 010 & 110 &001& 101 & 011 & 111\\[-1mm]
400 \; {\bf  \pi}& 100 &50 &40 &10&10&40&50&100\\[-1mm]
\text{\!\!log-lin. interaction } {\lambda'}:&\!\! 3.63& 0& 0& -0.17& 0& 0.52&0.63&0
\end{array}\right]
$$
To understand the examples the following matrices with marginal and partial correlations  are given for the three examples. The variables  are ($A,B,O$) in this  order. The matrices  show  correlations, $\rho_{st}$,   for $(s,t)=(1,2),(1,3), (2,3)$ in the lower triangle and partial correlations, $\rho_{st.v}=-\rho^{st}/\sqrt{\rho^{ss}\rho^{tt}}$, in the upper triangle; $v$ denotes the remaining variable and $\rho^{st}$ is a concentration, that is an element in the inverse covariance  matrix.
{\small
$$    
\begin{array}{ccc}\hline
                           1 & 0 &  0.51 \\
                           0.36 & 1& 0.51 \\ 
                            0.60 &  0.60  & 1\\  \hline
                            \end{array} \quad  
                            \begin{array}{ccc} \hline
                            1 & -0.18 &  0.35 \\
                           0 & 1 & 0.52\\
                     0.30   &  0.50   &      1\\ \hline
                         \end{array} \quad
                         \begin{array}{ccc} \hline
                         1 & -0.13 &  0.41 \\
                           0.10 &1\n & 0.50 \\ 
                            0.40   &  0.50    &  1\\ \hline
                         \end{array} 
$$        
}   
Note that in Example 3.1, we have $\rho_{12}=\rho_{13}\rho_{23}$ and 
in Example 3.2,  $\rho_{12.3}=- \rho_{13.2}\rho_{23.1}$.                      

Though these situations may be surprising when one sees them for the first time, they have simple explanations.  The strong version of the Simpson's paradox in Example 3.4 results for pair $(A,B)$,  say, when there are 
substantial dependences $A\pitchfork B |O$ and $A\pitchfork O|B$, and a strong dependence  $B\pitchfork O$; see \cite[sec. 6]{wer1987}. 

The weak versions are due to a dependence-inducing property.
This  property  is shared by joint Gaussian distributions; see e.g.  \citet[lem. 2.1]{wercox1998} and is  known as  \textit{singleton  transiti\-vity}  when it  is used for graphs representing a large class of graphical Markov models; see \citet[eqs. (10), (11)]{wer2015}.  The property has been studied for  binary variables by \citet[sec. 11]{simpson51}  and \citet{birch1963}, discussion of Equation 5.1, in the form 
\begin{equation} (A\ci B \mid  O  \text{ and } A\ci B) \text{ implies } (A\ci O \text{ or }  B\ci O). \label{eq:singleton}\end{equation}
Equivalently, equation \eqref{eq:singleton} is formulated as dependence-inducing with
$$  (A \pitchfork  O \text{ and }  B\pitchfork O)  \text{ implies }  (A\ci B| O  \text{ or } A\ci B) \text{ but not both},$$
and it applies to triples of variables also when a common conditioning set is added to each statement.
 For  variables $A,O, B$ this shows as in Fig.~\ref{fig:V} in graphs which are {\sf V}s :
 
\begin{figure}[H]
 \includegraphics[scale=0.4]{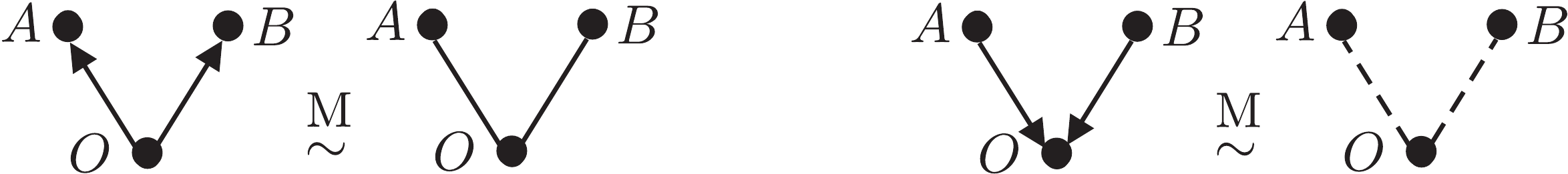}
\caption{\label{fig:V} Pairs of \protect{{\,\sf V}}s for binary $A,B,O$ with non-vanishing dependences associated to each edge present; left: $A\ci B|O$ 
and $A \pitchfork B$ is represented by a source {\sf V} and  by its Markov-equivalent, concentration-graph {\sf V}; right: $A\ci B$ and $A \pitchfork  B|O$ is represented by a sink {\sf V}and by  its Markov-equivalent, covariance-graph {\sf V}.}  
\end{figure}

In our  examples for $(A, B, O)$, the  log-linear interaction vector $\b \lambda$ tells  for Example 3.2 that the concentration graph 
is a {\sf V} since $\lambda_{12}=0$; for Examples 3.3 and 3.4 that it is a complete graph since all  two-factor terms are nonzero. On the other hand, the correlation matrices tell for Example 3.3 that  the covariance graph is a {\sf V}  since $\rho_{12}=0$.

The two weak versions of the Yule-Simpson paradox show that a distribution may have more independences than those displayed 
in its graphical representation. This has also been called its lack of faithfulness to the graph; see \citet{spirglysch2000}.

\subsection{Combination of independences} 
For three binary variables, the combination of independences was first studied by \citet[sec. 5]{birch1963}: 
$(A\ci B|O$ and $A\ci O|B)\implies A\ci BO$, but $(A\ci B$ and $A\ci O)\implies A\ci BO$ holds only when the three-factor interaction is lacking. 

More generally, if independences combine both downwards and upwards in a distribution as explained below, then the complete independence, $a\ci bc$, has several equivalent decompositions. Let $a,b,c$ be disjoint subsets of $\{1, 2, \ldots, d\}$ for the
random variables $X_1, X_2,\ldots X_d$, where each of $a,b,c$ contains at least one element, then 
 \begin{equation} a\ci bc \iff (a\ci b|c \text{ and } b\ci c) \iff (a\ci b \text{ and } a\ci c) \iff  (a\ci b|c \text{ and } a\ci c|b).\label{eq:compint}\end{equation}
 In all probability distribution,  the first equivalence holds but the third and fourth statement are only implied by $a\ci bc$. If the third statement implies $a\ci bc$,
 the independences combine upwards and if the fourth statement implies $a\ci bc$, they combine downward.

 Equation \eqref{eq:compint} is a known property of joint Gaussian distributions; see  \citet[def. 1]{matus2007}, and \citet[app.~2]{MarWer2009}, and  of \textit{traceable regressions}; see \citet[cor.~1]{wer2012}.  In the context of  graphical Markov models,  the upward and downward combination of independences  are called the composition and the intersection property, respectively.  Both are also properties of all currently known probabilistic graphs; see \citet{sadlau2014}.

A general sufficient condition for the downward  combination in discrete distributions are strictly positive 
probabilities, which are assumed in this paper.  In contrast to  Gaussian distributions, in palindromic Bernoulli distributions  independences  need not combine upwards. An extreme form of this is in Example 3.5, where the correlation matrix is the identity matrix, even though the four variables are dependent.
  $$\small
\begin{array}{r}
\hspace{-1mm} \text{\bf Example 3.5}\!\!\!\\[-1mm]
\\[-1.5mm]
 \end{array}
\left[\begin{array}{r r r r r r r r r r} 
 \text{levels  of  } A_1, A_2, A_3, A_4: & 0000 & 1000 & 0100 & 1100 &0010& 1010 & 0110 & 1110\\[-1mm]
 880\; \pi:&100& 10 & 10 & 100& 10 &100&100   &10\\[-1mm]
\end{array}\right]
$$
Here, the independences do not combine upwards since, for instance, both $A_1A_2 \ci A_3$ and $A_1A_2 \ci A_4$ are satisfied, but 
$A_1A_2 \pitchfork A_3A_4$ holds instead of $A_1 A_2 \ci A_3A_4$.
The reason is the reciprocal behaviour of the conditional odds-ratios which implies
that the  only nonzero log-linear interaction  is the four-factor term and that the 
conditional correlations vary with the levels of the third variable and can therefore not coincide with a partial correlation.

Similarly, the equality of conditional correlations may get destroyed with special covariance structures even when there is essentially no log-linear four-factor interaction. In Example 3.6, this happens with a \textit{funnel graph} which generalises the sink {\sf V} to more than two uncoupled nodes pointing to 
a common response; see \cite{LupMarBer2009} for estimation in such covariance graph models.
$$
\small \begin{array}{r}
\hspace{-1mm} \text{\bf Example 3.6}\!\!\!\\[-1mm]
\\[-1.5mm]
\end{array}
\left[\begin{array}{r r r r r r r r r r} 
\text{levels  of  } A_1, A_2, A_3, A_4: & 0000 & 1000 & 0100 & 1100 &0010& 1010 & 0110 & 1110\\[-1mm]
888\; \pi:&87 &24  & 102 & 9 & 60& 51 &87   &24\\[-1mm]
\end{array}\right]
$$
The importance of this type of palindromic structure is that it models studies using 
a  $2^k$ factorial design with equal allocation of the study subjects to all level 
combinations and special  sampling so as to get equal chances of success and 
of failure for a  binary response.

%
%
%
%
%
%
%
%
%

\section{Some special cases}
We discuss now palindromic distributions  arising when a Gaussian 
distribution is median dichotomized. Further, some details concerning 
the specification and estimation of undirected graphical models are given.

\subsection{Median dichotomization}
Let $(X_1,X_2)$ have a joint distribution function $F_{12}$ with
marginal distributions functions $F_1$ and $F_2$. Let further $U_1 = F_1(X_1)$ and $U_2 = F_2(X_2)$ be the
probability integral transforms of $X_1$ and $X_2$, so that $U_1$ and $U_2$ are
uniform.  Also let $\tilde X_j$ be the medians of 
$X_j$, $j=1,2$.
Consider now the median dichotomized variables, \begin{gather}
A_1 = \mathbb{I}[U_1 > \half], \quad A_2 = \mathbb{I}[U_2 > \half]
\label{meddic}
\end{gather}
where $\mathbb{I}[\cdot]$ is  the indicator function.  Then, the joint
distribution of $A_1$ and $A_2$ is a bivariate palindromic Bernoulli distribution, as given
in Section~\ref{sec:intro},
with $ \alpha = P( U_1 >  \half, U_2 > \half)$.

The variables $D_1 = (-1)^{A_1}$ and $D_2 = (-1)^{A_2}$, 
taking values $1, -1$,  have mean zero and unit variance,
so that $\xi_{12}= E(D_1D_2)$,  the  correlation coefficient between
$D_1$ and $D_2$, becomes the \emph{cross-sum difference} of the  joint
probabilities
\begin{gather}
\xi_{12} = 2\alpha - 2\beta = 4\alpha - 1.
\end{gather}
Thus the  correlation between two binary variables, which is
a multiple of the cross-product difference,
coincides in a bivariate palindromic Bernoulli distribution with the cross-sum
difference. This was not noted,
when   $\xi_{12}$ was proposed  as a measure of dependence between 
any two 
random variables $X_1$ and $X_2$ by \citet{blomqvist1950}:
\begin{align*}
4 \alpha -1 &= \pr\big\{(X_1 - \tilde X_1)(X_2 - \tilde X_2)> 0\big \} - 
\pr\big \{(X_1 - \tilde X_1)(X_2 - \tilde X_2)<0 \big\} \\
&= 2 \pr(U_1 \le   \half, U_2 \le  \half ) + 2 \pr(U_1 >   \half, U_2 >  \half )-1.
\end{align*}
\begin{rem}
The probability $\alpha$ may be interpreted as  the copula $C(\half, \half)$ of the random vector $(X_1,X_2)$, where
the function $C(u,v) = \pr(U_1\le u, U_2 \le v))$,  $0 \le u \le 1$,
$0\le v \le 1$.
\end{rem}
With the linear interaction expansion of equation \eqref{eq:lin},  the 
distribution of $D_1, D_2$ is 
\begin{gather}
P(D_1 = i, D_2 = j) = \textstyle \frac{1}{4}(1 + \xi_{12} ij), \qquad i,j = 1,
-1.
\end{gather} 

After median-dichotomizing $d>2$ continuous variables, the resulting
binary variables
$A_v, v = 1, \dots, d$ are still marginally uniform, but their joint
distribution is palindromic only   for centrally symmetric
variables, that is when
$X_v - \tilde X_v$ has the same distribution as $ -(X_v - \tilde X_v)$ for each
$X_v, v = 1, \dots d$. 

With $d = 3$,  
the joint distribution of the  median-dichotomized variables is  palindromic
with parameters $\alpha, \beta, \gamma$ and $\delta$,  as given
in section~\ref{sec:intro}. Their marginal  correlations are
$$
\xi_{12} = 4\alpha + 4 \delta -1, \quad 
\xi_{13} = 4\alpha + 4\gamma -1, \quad \xi_{23} = 4\alpha + 4\beta -1
$$
and the joint probability distribution is, with $i,j,k = 1, -1$. 
$$
P(D_1 = i, D_2 = j, D_3 = k) = \textstyle \frac{1}{8}( 1 + 
\xi_{12} ij + \xi_{13} ik + \xi_{23}{jk}) .
$$

\begin{exa}\label{ex:orth}
The following example gives the orthant probabilities of a trivariate, mean-centred Gaussian 
distribution having equal  correlations: $ -1/2 < \rho < 1$. The joint probability vector of the median-dichotomized variables is:
 $$8  \pi =(1 + 3\xi,  \quad 1-\xi, \quad 1-\xi,  \quad  1-\xi, \quad 1-\xi, \quad 1-\xi,  \quad  1-\xi, \quad 1 + 3\xi) $$
and the explicit transformations between the three types of parameters result with
\begin{equation}
 \xi = \frac{2}{\pi} \arcsin \rho, \quad 
\lambda = \frac{1}{4} \log  \frac{1+ 3\xi}{1-\xi}, 
\quad \eta = \mathrm{atanh} \,\xi. \label{eq:explicit}
\end{equation}  
The arcsin transformation is due to \cite{sheppard1898} and the obtained   distribution is  a \emph{concentric ring model}; see \cite{wermarzwier2014}.
\end{exa}

\begin{prop}
If $X$ has a $d$-variate Gaussian distribution with mean zero and 	correlation matrix $  R = [\rho_{st}]$ for $s, t = 1, \dots, d$ and  $A$ is the binary random vector obtained by median dichotomizing $X$, with linear interaction parameters $\xi_b$, then $  R$ can be reconstructed from the   correlation matrix $  R_A =[\xi_{st}]$ between the binary variables $A$  by
$$
\rho_{st} = \sin\{ (\pi/2) \xi_{st}\},\quad  s, t = 1, \dots, d.
$$ 
\end{prop}
The proof results by inverting the arcsin transformation of the quadrant probability
$$
\xi_{st} = 4\,\pr(X_s\le 0, X_t\le 0) -1 = 2\pi^{-1}\arcsin \rho_{st}.
$$  
As a consequence one may  reconstruct the original correlations 
from the palindromic Bernoulli distribution derived via the orthant probabilities.

\subsection{Maximum likelihood estimation}
 For  a palindromic Bernoulli distribution, given a random sample of size $n$, one has 
as  counts, that is as  observed cell frequencies: $n(a)$, $a \in \mathcal{I}$. The likelihood is
\begin{equation}
\textstyle{\prod_{a \in \mathcal{I}}} p(a)^{n(a)} =\textstyle{ \prod_{a \in \mathcal{I}_{0}}}
p(a)^{n(a)} p(\sim a)^{n(\sim a)}
= \textstyle{\prod_{a \in \mathcal{I}_0}}  p(a)^{n(a) + n(\sim a)}
\end{equation} 
where $\mathcal{I}_0$ is the set of half of the cells $a$ such that
$a_1 = 0$.
The sufficient statistics are thus the set  of the $2^{d-1}$ frequencies $n(a)
+ n(\sim a)$, obtained by summing each cell and its complement image. 
The maximum likelihood estimate  of a cell probability (or of a cell count) 
is the average of the two proportions (or of counts): 
\begin{equation}
\hat{p}(a) = \{n(a) + n(\sim a)\}/ (2 n), \quad 
\hat{n}(a) = \{n(a) + n(\sim a)\}/2. \label{eq:mle}
\end{equation}  
This produces a symmetrized vector of counts. 
 
For palindromic Bernoulli distributions, Wilks' likelihood ratio test statistic is  
\begin{equation}
w = 2 \;{ \textstyle{\sum_{a \in \mathcal{I}}}} n(a) \log \left( \frac{2\,n(a)}{n(a) +
n(\sim a)}\right).
\end{equation}
It has an asymptotic $\chi^2$ distribution with $2^{d-1}$ degrees of freedom;
see \citet[app. C]{edwards2000}.
The maximum likelihood estimates  of the linear interaction parameters 
are 
$$
\hat \xi_{b} = \begin{cases}
 0 & \text{if } |b| \text{ odd}, \\
 \sum_{a \in \mathcal{I}}  (-1)^{a \inner b} \, n(a)/n. &     \text{if } |b|
\text{ even}.
 \end{cases}
$$

Thus, the estimated   $\hat \xi_b$, for $|b|$ even, matches the observed moment statistic   
and for $|b|$ odd is zero. 
 For $|b|=2$, the  estimated marginal correlation, $\hat \xi_{12}$  coincides with the correlation coefficient in the fitted table 
$\hat{p}(a)$, hence is a cross-sum difference of the counts
\begin{equation}
\hat \xi_{12} = (n_{00}+ n_{11})-  (n_{01}+ n_{10}) \label{eq:est}.
\end{equation}
Since the log-linear and the multivariate logistic parameters are in a one-to-one relation to the linear interactions, the maximum likelihood estimates of their parameter vectors, result by the same transformations that hold for the  parameters; see \cite{fisher1922}. For the special transformations that apply here, see equations \eqref{eq:pifromlambda}, \eqref{eq:lin}.

In the following,  we speak of  maximum likelihood estimates  simply as `estimates'. Estimates may simplify further, when the distribution satisfies independence constraints in such
a way that they lead to a graphical Markov model; see, for an overview  of these models, e.g., \cite{DarLauSpe1980}, \cite{haberman1973}, \cite{wer2015}. 
\begin{exa}[A Markov chain] Let $A_1, A_2$ and $A_3$ be three binary random variables where $A_1$ and $A_3$ are conditionally independent
given $A_2$, so that the probabilities satisfy
$$p_{ijk} = p^{-1}_{+j+}\, p_{ij+}\,p_{+jk} \text{ for } i,j,k = \pm 1 \,.$$
Its undirected graph, called a concentration graph,  $1 \ful 2 \ful 3$, has a missing edge for nodes 1 and 3, representing $A_1$ and $A_2$, and it is  a  simplest type of a graphical Markov model, a  Markov chain in 3 variables; see also Example 3.1.
 
The  log-linear parameters are  constrained by
$\lambda_{13} = \lambda_{123} = 0$ for the conditional independence of pair $(1,3)$. If, in addition,   the distribution is palindromic,  the odd-order parameters are zero so that also $\lambda_1 = \lambda_2 = \lambda_3 = 0$.
In general Bernoulli distributions, the  minimal sufficient statistics are the observed counts corresponding to the \emph{cliques} of the graph,
i.e., the  maximal complete subsets of the nodes,  here just the node pairs (1,2) and (2,3). 
However for a palindromic Bernoulli distribution, the minimal sufficient statistics 
are the estimated counts   $\hat n_{ij+}$ and  $\hat n_{+jk}$ for margins $(1,2)$ and $(2,3)$, defined as in equation \eqref{eq:mle} from the symmetrized table, so that
$$ 
\tilde p_{ijk} = n^{-2}\,\hat n_{ij+} \hat n_{+jk}\text{ for }  i,j,k = \pm 1. 
$$
where $\tilde p_{ijk}$ is the estimate of $p_{ijk}$ under the Markov chain model. 
\end{exa} 

This example illustrates  how independence constraints,  conditionally given all remaining variables, 
simply add to the linear constraints on canonical parameters of a palindromic Bernoulli distribution. Moreover, when the model is decomposable, since its concentration graph is chordal, see e.g. \cite{DarLauSpe1980},  it can be generated  by a  linear triangular system; see end of Section~\ref{sec:charact}. 

As far as the maximum likelihood estimation of palindromic graphical models is concerned, 
the hierarchical constraints of conditional independence and the non-hierarchical constraints of central symmetry are well compatible with one another. Thus one can fit a given graphical model to the symmetrized counts or equivalently symmetrize the fitted counts 
under the model.

\section{Palindromic Ising models}\label{sec:ising}

 We now introduce palindromic Ising models, especially when combined with conditional independence constraints.  Ising models are joint Bernoulli distributions without any higher than  two-factor log-linear interactions. An Ising model  is palindromic if it has also  uniform margins. As mentioned before, this leads for the  $(-1,1)$ coding to  binary variables which have  zero means, unit variances and  covariances coinciding with Pearson's correlations.    
 General Ising models have, for instance, been studied  as lattice systems,   \cite{besag1974}, and  as binary quadratic-exponential distributions, \cite{coxwer1994b}.

The concentration graph of an Ising model in $d$ variables has $d$ nodes and at most one undirected edge coupling a node pair.   The  edge  $(i,j)$   in this 
concentration graph is missing  if the  two-factor log-linear interaction  of pair $(i,j)$ vanishes. Each missing edge $(i,j)$ means $i\ci j|\{1,\ldots, d\}\setminus\{i,j\}$.  
Simpler independence statements  such as  $i\ci j|C$, for $C \subset \{1, \ldots, d\}\setminus \{i,j\}$  result  if every path
between $i$ and $j$ has a node in $C$; see for instance  \cite{DarLauSpe1980}. Recall that  nodes and edges of the cliques of the graph form its maximal complete subgraphs, that is those node subsets without any missing edge which become incomplete when one more node is added.

\begin{prop} If the concentration graph of a palindromic Bernoulli distribution has
largest clique size three, then it is a palindromic Ising model. \end{prop}

\begin{proof}
An unconstrained   palindromic Bernoulli distribution has a  complete concentration graph in nodes $\{1, \dots,d\}$.
Each removed edge introduces an independence  and reduces  the size of a  generating clique. If the largest
clique-size is three, then each  of the $2^3$ generating probabilities   for the  joint palindromic distribution 
 has  no  3-factor log-linear interaction. 
  \end{proof}
  
\begin{prop}
Palindromic Ising  models are closed under marginalizing for $d\leq 4$.	
\end{prop}
\begin{proof}
In an Ising model, all higher than 2-factor log-linear interactions are zero and all trivariate and bivariate 
marginal distributions are Ising models because palindromic distributions are closed under marginalizing.
\end{proof}

 For palindromic Ising models of more than four variables, cliques of even order may get  induced by marginalizing together 
 with corresponding even-order log-linear interactions. This happens, for instance, if the concentration graph is a star graph with edges
$is$, for $i = 1, \dots, 4$,  and marginalizing is over the common source $s = 5$.  Notice however that, in any case, the bivariate  
and trivariate marginal probabilities remain palindromic Ising models.

\begin{prop} 
For a decomposable palindromic Ising model  with largest clique size three, the maximum likelihood estimates are obtained in closed form from the marginal correlations in the symmetrized $2\times 2$ tables within its 2-node and 3-node cliques.  \label{thm1} 
\end{prop}
\begin{proof} 
If a model is decomposable, then an ordering of the cliques $C_t$, $t = 1, \dots, T$ can be found, such that the joint probability factorizes and estimation simplifies; see \cite{sundberg1975}.
This ordering satisfies the \emph{running intersection property} meaning that the 
sets $S_1 = C_1 \cap (\cup_{t>1} C_t)$, $S_2 = C_2 \cap (\cup_{t>2} C_t),$ \dots, $S_{T-1} = C_{T-1} \cap C_T$, called separators, are all complete, that is all nodes are coupled by an edge. Then, the joint probability $p(a)$ factorizes into the  product of the marginal distributions over cliques $C_t$
divided by the product of the the marginal distributions over the separators $S_t$. 

As cliques and separators have largest size $\le 3$, and the associated marginal distributions  are palindromic Ising models, they are fitted in closed form directly from the marginal $2\times 2$ tables of the symmetrized counts by using for instance the marginal correlations of such tables, that is the cross-sum differences.       
\end{proof} 

Example 5.1 below gives three non-decomposable palindromic Ising models having a  
so-called chordless four-cycle.  They have maximal clique size two and edges for
$(1,3)$, $(1,4)$, $(2,3)$, $(2,4)$.  They differ in that all nonzero log-linear interactions are positive in the first and negative in the second case. In the third one,  there is a chordless cycle in the concentration graph as well as in the covariance graph, that is not only the two independences of the concentration graph hold, $1\ci 2|34$ and $3\ci 4|12$, but, in addition, $1\ci 2$ and  $3\ci 4$. 
 $$\small
\begin{array}{r}
\hspace{-1mm} \text{\bf Example 5.1}\!\!\!\\[-1mm]
\\[-1.5mm]
 \end{array}
\left[\begin{array}{r r r r r r r r r r} 
 \text{levels  of  } A_1, A_2, A_3, A_4: & 0000 & 1000 & 0100 & 1100 &0010& 1010 & 0110 & 1110\\[-1mm]
 336\; \pi:&75& 15& 15 &3 & 15 & 15 &15& 15\\[-1mm] 
 336 \; \pi:&3& 15 & 15 & 75 & 15 & 15 &15 &15      \\[-1mm] 
 336\; \pi:&35 & 35& 7 & 7  &  7& 35 &7& 35\\[-1mm] 
\end{array}\right]
$$
In these chordless cycles, the missing edges in 
the concentration graph show also as zeros in the matrix of partial correlations
given all remaining variables that is  in $-\rho^{ij}/\sqrt{\rho^{ii}\rho^{jj}}$.
This points to the possible extension of  a result by  \cite{lohwain2013}  to 
include chordless cycles for palindromic Ising models.
\begin{prop} There is no effect reversal in a palindromic Ising model if all its nonzero log-linear interactions are positive.
\end{prop}
This result is a direct consequence  of Proposition 3.4 $(ii)$
in \cite{FalLauEtal2016} for totally positive  palindromic Ising models. We conjecture that the same 
holds, when all nonzero log-linear interactions in a palindromic Ising model are negative, such as
in our second case  of Example 5.1.

\section{A case study}\label{sec:case}

The following case study illustrates some of the obtained results. 
For a sample of grades obtained at the University of Florence, we aim  at predicting grades in  Physics in terms of 
given grades in Algebra, Analysis and Geometry. The passing grades range in each subject from 18 to 31. 
We  use sums of grades over  exams  in three successive years and have data for
$n=78$ students who reached in each of the subjects a sum of at least 60 points. Instructors expect  positive 
correlations for each pair of these grades and no sign reversal for the correlations at fixed level
combinations of the other variables. The data are  in Appendix 2.

The four summed grades are closely bell-shaped, each of their scatter plots shows a nearly
elliptic form as well as the plots of residual pairs obtained after linear least squares regression of  each grade on the other three.
There is also no evidence for nonlinear relations in the  probability plots of  \cite{coxwer1994}. Thus, there is
substantive and empirical support for assuming  a  joint Gaussian
distribution.
   
After replacing  for pairs (1,4) and (2,4)  the observed correlations by $\hat{ r}_{14}=r_{13}r_{34}$, $\hat{r}_{24}=r_{23}r_{34}$, in Table \ref{tab:concex1}
we have  the estimate of the correlation matrix, which has  zeros for pairs (1,4) and (2,4) in its inverse,  in its concentration matrix; see
e.g. \cite{wermarcox2009}, equation (2.8). 

\begin{table}[htp!]
\centering 
\caption{For four fields and 78 students, observed marginal correlations, $r_{ij}$ (below the diagonal), concentrations  on  the diagonal 
and partial correlations, $r_{ij.kl}$ (above the diagonal).}\label{tab:concex1}
\begin{tabular}{lrrrr}
\hline
 & Analysis & Algebra & Geometry &Physics\\ \hline
1:=Analysis & $  2.64  $ & $0.27$ & $0.34$ & $0.17$  \\
2:=Algebra & $  0.72$ & $3.03$ & $0.51$ & $0.04$ \\
3:=Geometry &  $0.76$&  $0.80$& $4.07$ & $0.38$ \\
4:=Physics &   $0.62$& $0.60$& $0.71$&   $2.09$ \\ \hline
\end{tabular}
\end{table}

Wilks' likelihood-ratio test statistic on 2 degrees of freedom shows with a value of $w = 2.8$ a good fit
to the model with  generating sets $\{\{1,2,3\}, \{3,4\}\}$. This implies conditional independence of the grade in Physics from those in Analysis and Algebra given the grade in Geometry. This follows directly,  for instance, with the corresponding concentration graph, on the left of Fig.~\ref{fig:conc}. It has 
the cliques   $\{1,2,3\}$ and $\{3,4\}$ for which  node  3 separates node 4 from nodes 1,2 since to reach nodes 1,2 from node 4, one has to pass via node 3.

Similarly, after replacing the marginal correlations for pairs $(1,2)$, $(1,3)$ and $(2,3)$ by their average $\hat{r}=0.76$,  we have for the submatrix of (1,2,3)
the conventional estimate of an equicorrelation matrix; see \citet[Section 3]{olkinpratt1958}. This is here well-fitting since $w = 3.4$  on $2$ degrees of freedom.  The grade in Physics, correlates with this sum score as $0.706$, even slightly less than  with the grade  in Geometry alone, where $r_{34}=0.709$. This is plausible in view of the well-fitting  Markov structure. 
\begin{figure}[htp!]
\centering
\includegraphics[scale=.38]{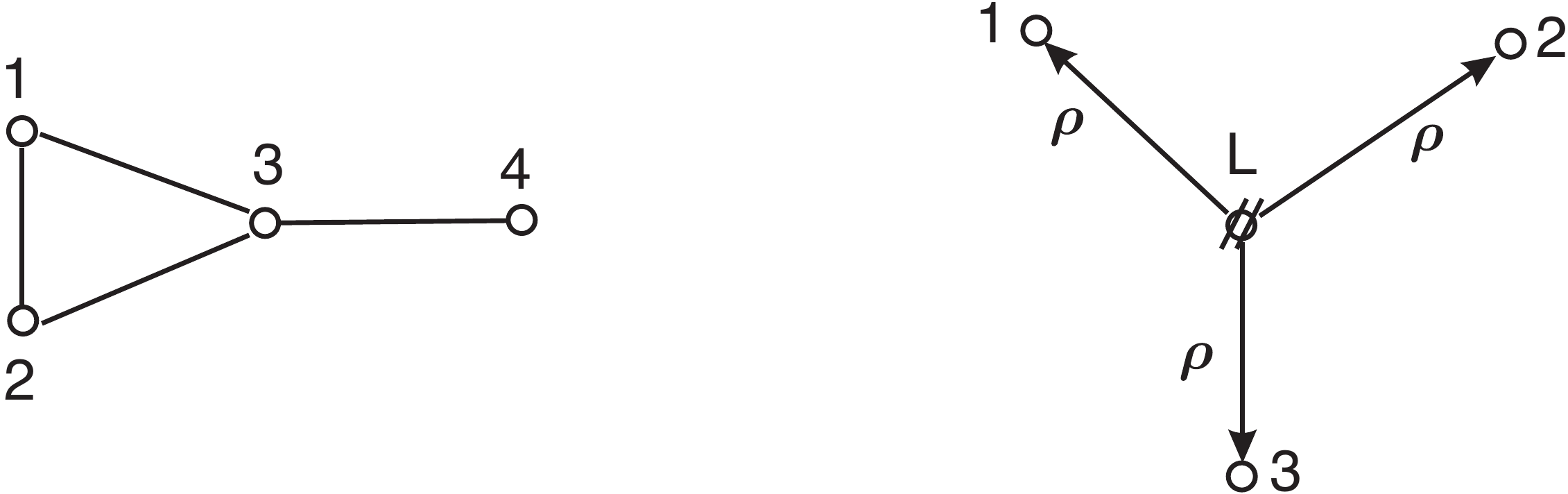}
\caption{Left: the well-fitting concentration graph for the Florence grades;  right: a possible generating graph for grades $1,2,3$.}
\label{fig:conc}
\end{figure}

A possible generating graph for the Gaussian equicorrelation matrix is the star graph displayed on the right of  Fig.~\ref{fig:conc}. In it 
mathematical ability is represented by the unobserved inner node, $L$, and the three  grades are the  outer nodes of the graph, shown as responses
to $L$ by arrows starting at $L$ and pointing to the uncoupled nodes 1,2,3; each arrow has assigned to it the same positive correlation $\rho$. After marginalizing over $L$,  each outer pair is correlated like $\rho^2$. We shall see next how well  these results are reflected in the dichotomized  data.

After median-dichotomizing the grades with  jittering, we generate precisely uniform binary variables, the marginal distributions of which  differ only little from  those obtained by simple median-dichotomizing.  One obtains the estimate of a palindromic contingency table in closed form using equation \eqref{eq:mle}
and as we shall see, the same well-fitting concentration graph as on the left of Fig.~\ref{fig:conc}.

The observed contingency table is given next,  together  with, in the additional rows  and in the following order,  the estimates of  palindromic counts,  and of the counts after imposing, in addition, conditional independences for the model of Fig.~\ref{fig:conc} on the left, and the estimates of  the corresponding log-linear interactions.

\begin{table}[htp!] 
\caption{Cells $ijkl$, levels of interactions, counts $n_{ijkl}$, estimates  of  palindromic counts, of the palindromic concentration graph counts and  of the log-linear interactions under the latter model} \label{tab:conc}
$$  
\setlength{\arraycolsep}{.5\arraycolsep}
\begin{array}{rrrrrrrrrrrrrrrr} \hline
 0000 & 1000 & 0100 & 1100 & 0010 & 1010 & 0110& 1110&  0001 & 1001 & 0101 & 1101 & 0011 & 1011 & 0110& 1111 \\
 \emptyset & 1 & 2 & 12 & 3 & 13 & 23 & 123 &4& 14 & 24 & 124 & 34 & 134 & 234 & 1234\\ \hline
 22   &3     &      3      & 0     &1        &  0     &     1&          9& 6&        2&          2&          1&          3&         2&   1&    22\\
 22   & 2&      2.5& 1.5  &1  &1  &         1.5 &     7.5 &7.5 &          1.5&         1.0&        1.0&        1.5&          2.5&        2.0   &     22.0\\
 21.2  &  2.5 &  2.5& 1.8&  0.7 & 1.0 &  1.0  &  8.3 &8.3&         1.0&         1.0&          0.7&        \footnotesize  1.     8&    2.5&        2.5&         21.2\\
  0.90 &  0   &  0   & 0.45   &0  & 0.62 &  0.62  &  0 & 0      &        0&             0&      0&   0.47&             0 &         0&         0\\ \hline
\end{array}
$$  
\end{table}
The palindromic concentration graph model fits well, with $w=10.3$ on $11$ degrees of freedom. This  decomposes into $w=9.1$ on $8$  degrees of freedom for the saturated
palindromic model and $w=1.2$ on $3$ degrees of freedom for the additional independence constraints.

Values of the studentized log-linear interactions are $2.5$, $3.5$,
$3.5$ and $3.7$ for $\lambda_{12}$, $\lambda_{13}$, $\lambda_{23}$ and $\lambda_{34}$, respectively.
Thus, the same independences as for the underlying joint Gaussian distribution fit also the median-dichotomized data and further
simplifications are not compatible given the sizes of the remaining studentized interactions. The partial  correlations  implied
by the well-fitting palindromic Markov structure has also zeros  for pairs $(1,4)$. 

The sum score of the median-dichotomized grades 1,2,3 leads  as in the underlying Gaussian distribution  not to an improved prediction
of grades in Physics. To our starting question, we get two summarizing answers.  Given a grade below the median in Geometry,
one predicts that 72\%  of these students will  have a grade below the median in Physics  and, similarly,  
 given a grade above the median in Geometry,
one predicts that 72\%  will  have a grade above  the median in Physics.

\section{Discussion} We say that centrally-symmetric  Bernoulli distributions  are  palindromic 
since their probabilities, at the fixed level of one of the variables repeat in reverse order for the  second level of this variable and
thereby mimic palindromic sequences of characters as introduced  in linguistics.

A palindromic Bernoulli distribution is characterized by the vanishing of all odd-order log-linear interactions. Hence, such zero constraints
 lead to a non-hierarchical,  log-linear model which give  centrally symmetric probabilities. Until now, it was 
only known that in centrally-symmetric Bernoulli distributions, all odd-order log-linear parameters vanish; see \citet[app. C]{edwards2000}. With these linear constraints,  distributions result which are in the regular exponential family.

Palindromic Bernoulli distributions may also  be parameterized with all odd-order interactions vanishing in a linear-in-probability model and in a multivariate-logistic model. The parameters in the three types of model are in  one-to-one relations; see Section 2. These relations are now available  in closed form for the linear and the log-linear formulations, while in general, iterative procedures are needed when the multivariate logistic formulation is involved. In any case,   equivalent parameterizations assure that the maximum-likelihood estimates of the parameters are in the same one-to-one relation; see  \cite{fisher1922}.

It is remarkable that a palindromic Bernoulli distribution can be expressed precisely as a log-linear and as a linear model, since  log-linear parameters use the notion of multiplicative interactions and the linear-in-probability models are based  instead on the notion of additive interactions as discussed, for instance by \cite{darrochspeed1983}.  

The log-linear parameterization shows that positive palindromic
Bernoulli distributions are in the regular exponential family with and without additional  independence constraints in its concentration graph.
A palindromic Ising model may have only log-linear two-factor interactions as  non-vanishing canonical parameters, while in their  linear-model formulations higher-order interactions may  be present.  
The palindromic property is preserved under marginalizing over any subset of the variables; see Proposition~\ref{marginalizePBD}, 
even though one  may no longer have an Ising model after marginalizing
over some of the variables.

Another property is important for applications. In palindromic Bernoulli distributions, many
other measures of dependence of a variable pair  are one-to-one functions of the odds-ratio; in particular the relative risk, used mainly in epidemiology, and  the risk difference,   employed almost exclusively in the literature on causal modelling. Only if a measure of dependence is a function of the 
odds-ratio, it  varies independently  of its margins; see \cite{edwards1963} and only then,  measures of  bivariate dependence become directly comparable under different sampling schemes, for instance when the overall count is fixed as in  a cross-sectional study or one of
the margins is fixed as in  a prospective study or the other margin is fixed as in a retrospective study.

 We expect  that with a direct extension of the palindromic property to discrete variables of more levels, similar attractive properties can be obtained as for the palindromic Bernoulli distribution.

\section*{Appendix 1. Proofs} 

\begin{proof}[Proof of Proposition~\ref{prop:tri}]
The proof is  by induction. 
We know that $A_1$ has a palindromic distribution. For $s = 2, \dots, d$ we assume that the random vector
$A_{[s-1]} = (A_1, \dots, A_{s-1})$ has a palindromic distribution, and then we show that the distribution of $A_{[s]} = (A_1, \dots, A_{s})$ is palindromic. 
Let $\mathcal{I}_{\rm even}$ denote the subset of $\{0,1\}^s$ with even order
and split it in two parts  
$$
\mathcal{I}_{0} = \{ a \in \mathcal{I}_{\rm even} :  a_s  = 0\}, \;
\mathcal{I}_{1} = \{ a \in \mathcal{I}_{\rm even} :  a_s  = 1\}.
$$
We then start from the identity  
$$
\pr(A_{[s]} = a_{[s]}) = \pr(A_{[s-1]} = a_{[s-1]}) \pr(A_s = a_s\mid  A_{[s-1]} = a_{[s-1]}) 
$$
and after substituting equations \eqref{eq:lin} and \eqref{eq:maineffects} and taking into account that by assumption $A_{[s-1]}$ has a palindromic distribution and thus  $\xi_{b} = 0$ for all $b \in \mathcal{I}_1$, 
we have 
$$
\pr(A_{[s]} = a) = 2^{-s} \textstyle{\sum_{b \in \mathcal{I}_0}} \xi_{b} (-1)^{a \inner b} \cdot  
\{1 + \textstyle{\sum_{j = 1}^{s-1}} \beta_{sj} (-1)^{a \inner e_{s,j} } \}
$$
where $e_{s,j}$ is a binary vector of dimension $s$ with ones exactly in positions $s$ and $j$. After multiplying and collecting terms
    we get    with
\begin{gather}
\textstyle
\xi_b = \sum_{v\in I_0\colon v \bigtriangleup \{s,j\} = b}  \xi_v  \beta_{s,j}, \quad \text{ for } b \in \mathcal{I}_1,\label{eq:1}
\end{gather}
\begin{equation*}
\pr(A_{[s]} = a)  = 2^{-s} \Big(\textstyle{\sum_{b \in \mathcal{I}_0} }\xi_{b} (-1)^{a\inner b} +\textstyle{ \sum_{b \in \mathcal{I}_1} }\xi_b(-1)^{a \inner b} \Big),
\end{equation*}
where $\bigtriangleup$ denotes the  symmetric difference of sets. Therefore $A_{[s]}$ has a linear parameterization with exclusively
even order interactions and  hence is palindromic. Thus, by induction,  the distribution of $A_{[d]} = A$ is palindromic. 
From the recursive equation \eqref{eq:1},  each linear interaction is a linear function of the regression parameters $\beta_{s,j}$.
\end{proof}

\begin{proof}[Proof of Prop.~\ref{prop:triv}]
It is known that in a strictly positive $2^3$ table  $A_1 \ci  A_2 \mid A_3\iff (\lambda_{12} = 0$ and $\lambda_{123} =0)$. As in a trivariate palindromic table the three-factor interaction is always zero the single condition $\lambda_{12} = 0$ is necessary and sufficient. This is in turn equivalent to a single condition 
    on the partial correlation $\rho_{12.3} =0$.  
  
With no constraints other than those of palindromic 
distributions, there is a smooth one-to-one transformation
$(\lambda_{12}, \lambda_{13}, \lambda_{23} ) \longleftrightarrow (\rho_{12}, \rho_{13}, \rho_{23})$ where the marginal correlations are the free parameters of $\b \xi$.

In addition, there is a one-to-one smooth transformation between the simple correlations and the partial correlations $(\rho_{12}, \rho_{13}, \rho_{23} )\longleftrightarrow (\rho_{12.3}, \rho_{13.2}, \rho_{23.1})$, since the  conditional and partial correlations coincide. Also, the function $g\colon \reals \rightarrow [-1, 1]\colon \lambda_{12}\mapsto r_{12.3}$ is strictly monotone increasing for any values of $\lambda_{13}$ and $\lambda_{23}$ and has a single zero in the origin,  so that the partial correlation $r_{12.3}$ has the same sign as the log odds-ratio $\lambda_{12}$.   

In more detail, the three variable table of Section 1,  supplemented by  both margins for the two conditional tables of $A_1, A_2$ can be written in terms of simple correlations as
\begin{center}
\begin{small}
\renewcommand*{\arraystretch}{0.8}
\begin{tabular}{lccccccc}
\hline
    \quad         &\hspace{-6mm}  $A_3\!:$\hspace{-5mm}  &0   & 0   &   & 1  & 1  &     \\
  $A_1 $         &\hspace{-6mm}  $A_2\!:$\hspace{-5mm}  &0   & 1  &\text{sum}       & 0  & 1  & \text{sum}\\ \hline
\, 0    & &$\alpha $         & $\gamma  $   &   $(1+\rho_{13})/4$ &$ \delta $  &$  \beta$   &    $(1-\rho_{13})/4  $
\\ 
\, 1         & &  $\beta $        &$ \delta  $   &  $(1-\rho_{13})/4$  & $\gamma $  &$ \alpha $    &  $(1+\rho_{13})/4  $
\\ \hline
\text{sum} &&$(1+\rho_{23})/4\nn$      & $(1-\rho_{23})/4$\nn   & 1/2&$(1-\rho_{23})/4\nn $&
$(1+\rho_{23})/4$  \nn& 1/2         \\   \hline
\end{tabular}
\end{small}
\end{center}
The probabilities in each of  the $2^2$ tables, expressed with  margins and $\rho_{12|3}$ give  e.g.
$$
8\alpha=(\rho_{12}-\rho_{13}\rho_{23}) +  \{(1+\rho_{13} )(1+\rho_{23})\}, \quad \quad 8\delta=(\rho_{12}-\rho_{13}\rho_{23}) +  \{(1-\rho_{13} )(1-\rho_{23})\},
$$
since the product of all four margins is  in both tables $(1-\rho_{13}^2)(1-\rho_{23}^2)/16^2$ and 
$$
(\rho_{12}-\rho_{13}\rho_{23})/16= \{\fracshalf-(\beta+\gamma+\delta)\} \delta- \beta\gamma=\alpha\delta-\beta\gamma\, .$$
Thus, $\rho_{12|3}=\rho_{12.3}$ and $\alpha\delta-\beta\gamma=0$ if and only if $\lambda_{12}=0$.
The above $2^3$  table shows also directly that under the independence constraint $1\ci 2|3$, we have
\begin{eqnarray*}
\tilde{\alpha} &=& \{1+\rho_{13} )(1+\rho_{23} )\}/8, \quad
\tilde{\delta} =   \{(1-\rho_{13} )(1-\rho_{23})\}/8\\
\tilde{\beta} &=& \{1-\rho_{13} )(1+\rho_{23} )\}/8,\quad 
\tilde{\gamma} =  \{1+\rho_{13} )(1-\rho_{23} )\}/8
\end{eqnarray*}
so that e.g.  $\tilde{\alpha}\tilde{\delta}$ and the  probability  $\tilde{\alpha}+\tilde{\delta}$, induced in the marginal table of $A_1, A_2$  for $(1,1)$,  are
$$ 
\tilde{\alpha}\tilde{\delta}=(1-\rho_{13}^2)(1-\rho_{23}^2)/8^2,   \quad \quad  \tilde{\alpha}+\tilde{\delta}=(1+\rho_{13}\rho_{23})/4. \qquad\qedhere
$$ 
\end{proof}  

\section*{Appendix 2. The data for the case study}
The columns of the following table contain sums of grades of three exams in four subjects for $n=78$ mathematics students at the University of Florence.

\begin{table}[hh] 
\centering 
\caption{Summed grades over 3 exams in the order: Analysis, Algebra, Geometry and Physics.}\label{tab:concex2} 
\setlength{\tabcolsep}{3.5pt}
\small
\begin{tabular}{rrrrrrrrrrrrrrrrrrrrrrrrrrrrrrrrrr} 
\hline
78&78&74&80&88&77&79&85&82&82&74&89&85&77&93&85&79&85&74&69&78&88&67&92&85&69\\
\hspace{-2mm}76&75&71&77&81&79&77&90&79&72&62&90&75&83&92&88&80&88&68&80&75&88&70&89&88&75\\
\hspace{-2mm}82&81&74&71&85&74&81&83&73&73&71&86&84&84&93&82&78&90&70&79&71&89&68&91&91&62\\
\hspace{-2mm}85&77&80&80&79&80&75&82&71&71&72&87&82&69&90&75&75&82&70&78&72&77&69&93&87&68\\[2mm]
\hspace{-2mm}79&92&76&88&73&91&76&71&65&74&80&71&78&77&70&83&89&72&82&77&91&92&75&90&90&93\\
\hspace{-2mm}78&92&78&87&68&85&78&79&68&76&89&74&81&74&68&89&81&76&81&74&92&92&69&79&82&93\\
\hspace{-2mm}71&92&84&88&64&83&82&69&71&75&80&71&85&69&67&88&83&75&83&82&93&92&72&90&89&93\\
\hspace{-2mm} 79&90&86&78&69&75&82&71&63&72&78&74&81&67&66&72&82&75&79&76&92&87&75&79&78&89\\[2mm]
\hspace{-2mm} 92&87&81&82&76&86&92&87&79&91&88&90&90&92&89&83&77&69&89&92&86&76&68&79&76&88\\
\hspace{-2mm} 93&83&69&70&75&71&80&70&70&77&88&92&85&92&84&83&82&74&83&92&74&71&62&68&66&89\\
\hspace{-2mm} 93&87&74&67&80&69&87&77&69&92&83&91&82&91&86&83&80&83&83&90&78&71&65&74&83&91\\
 \hspace{-2mm} 89&77&81&79&84&72&80&81&70&79&77&72&88&81&86&81&78&76&77&79&73&69&69&72&80&85\\ \hline
\end{tabular}   
\end{table}

\section*{Acknowledgement}
We thank G. Bulgarelli, Univ. of Florence, for the data; G. Ottaviani and the referees for their helpful comments. We used  Matlab\textregistered, 2015, and
for exact  median-dichotomizing, `jitter.m'  by R. Cotton, who adapted  the R version by W.~Stahel and M.~Maechler.

\bibliographystyle{biometrika}

\end{document}